\documentclass{article}

\usepackage{amsmath}
\usepackage{amssymb}

\newtheorem{theorem}{Theorem}
\newtheorem{corollary}[theorem]{Corollary}
\newtheorem{conjecture}[theorem]{Conjecture}
\newtheorem{lemma}[theorem]{Lemma}
\newtheorem{claim}[theorem]{Claim}

\newtheorem{proposition}[theorem]{Proposition}

\newenvironment{proof}{\noindent\bf{Proof.}\rm}{\hfill$\blacksquare$\bigskip}

\begin{document}

\title{Tighter bounds for online bipartite matching}

\author{Uriel Feige~\thanks{Department of Computer Science and Applied Mathematics, The Weizmann Institute. {\tt uriel.feige@weizmann.ac.il}}}

\maketitle

\begin{abstract}
We study the online bipartite matching problem, introduced by Karp, Vazirani and Vazirani [1990]. For bipartite graphs with matchings of size $n$, it is known that the {\em Ranking} randomized algorithm matches at least $(1 - \frac{1}{e})n$ edges in expectation. It is also known that no online algorithm matches more than $(1 - \frac{1}{e})n + O(1)$ edges in expectation, when the input is chosen from a certain distribution that we refer to as $D_n$. This upper bound also applies to {\em fractional} matchings. We review the known proofs for this last statement. In passing we observe that the $O(1)$ additive term (in the upper bound for fractional matching) is $\frac{1}{2} - \frac{1}{2e} + O(\frac{1}{n})$, and that this term is tight: the online algorithm known as {\em Balance} indeed produces a fractional matching of this size. We provide a new proof that exactly characterizes the expected cardinality of the (integral) matching produced by {\em Ranking} when the input graph comes from the support of $D_n$. This expectation turns out to be $(1 - \frac{1}{e})n + 1 - \frac{2}{e} + O(\frac{1}{n!})$, and serves as an upper bound on the performance ratio of any online (integral) matching algorithm.
\end{abstract}

\section{Introduction}

Given a bipartite graph $G(U,V;E)$, where $U$ and $V$ are the sets of vertices and $E \in U \times V$ is the set of edges, a matching $M \subset E$ is a set of edges such that every vertex is incident with at most one edge of $M$. Given a matching $M$, a vertex is referred to as either matched or exposed, depending on whether it is incident with an edge of $M$. A maximum matching in a graph is a matching of maximum cardinality, and a maximal matching is a matching that is not a proper subset of any other matching. Maximal matchings can easily be found by greedy algorithms, and maximum matchings can also be found by various polynomial time algorithms, using techniques such as alternating paths or linear programming (see~\cite{LP86} and references therein). In every graph, the cardinality of every maximal matching is at least half of that of the maximum matching, because every matched edge can exclude at most two edges from the maximum matching.

For simplicity of notation, for every $n$ we shall only consider the following class of bipartite graphs, that we shall refer to as $G_n$. For every $G(U,V;E) \in G_n$ it holds that $|U| = |V| = n$ and that $E$ contains a matching of size $n$ (and hence $G$ has a perfect matching).  The vertices of $U$ will be denoted by $u_i$ (for $1 \le i \le n$) and the vertices of $V$ will be denoted by $v_i$ (for $1 \le i \le n$).  All results that we will state for $G_n$ hold without change for all bipartite graphs, provided that $n$ denotes the size of the maximum matching in the graph.

Karp, Vazirani and Vazirani~\cite{KVV90} introduced an online version of the maximum bipartite matching problem. This setting can be viewed as a game between two players: a maximizing player who wishes the resulting matching to be as large as possible, and a minimizing player who wishes the matching to be as small as possible. First, the minimizing player chooses $G(U,V;E)$ in private (without the maximizing player seeing $E$), subject to $G \in G_n$. Thereafter, the structure of $G$ is revealed to the maximizing player in $n$ steps, where at step $j$ (for $1 \le j \le n$) the set $N(u_j) \subset V$ of vertices adjacent to $u_j$ is revealed. At every step $j$, upon seeing $N(u_j)$ (and based on all edges previously seen and all previous matching decisions made), the maximizing player needs to irrevocably either match $u_j$ to a currently exposed vertex in $N(u_j)$, or leave $u_j$ exposed.

There is much recent interest in the online bipartite matching problem and variations and generalizations of it, as such models have applications for allocation problems in certain economic settings, in which buyers (vertices of $U$) arrive online and are interested in purchasing various items (vertices of $V$). For more details, see for example the survey by Metha~\cite{Mehta13}.

An algorithm for the maximizing player in the online bipartite matching setting will be called {\em greedy} if the only vertices of $U$ that it leaves unmatched are those vertices $u\in U$ that upon their arrival did not have an exposed neighbor (and hence could not be matched). It is not difficult to see that every non-greedy algorithm $A$ can be replaced by a greedy algorithm $A'$ that for every graph $G$ matches at least as many vertices as $A$ does. Hence we shall assume that the algorithm for the maximizing player is greedy, and this assumption is made without loss of generality, as far as the results in this manuscript are concerned.

Every greedy algorithm (for the maximizing player) produces a maximal matching, and hence matches at least half the vertices. 
For every deterministic algorithm, the minimizing player can select a bipartite graph $G$ (that admits a perfect matching) 
that guarantees that the algorithm matches only half the vertices. (Sketch: The first $\frac{|U|}{2}$ arriving vertices have all of $V$ as their neighbors, and the remaining $\frac{|U|}{2}$ are neighbors only of the $\frac{|V|}{2}$ vertices that the algorithm matched with the first $\frac{|U|}{2}$ vertices.)

To improve the size of the matching beyond $\frac{n}{2}$, Karp, Vazirani and Vazirani~\cite{KVV90} considered randomized algorithms for the maximizing player. Specifically, they proposed an algorithm called {\em Ranking} that works as follows. It first selects uniformly at random a permutation $\pi$ over the vertices $V$. Thereafter, upon arrival of a vertex $u$, it is matched to its earliest (according to $\pi$) exposed neighbor if there is one (and left unmatched otherwise).
As the maximizing algorithm is randomized (due to the random choice of $\pi$), the number of vertices matched is a random variable, and we consider its expectation.

Let $A$ be a randomized algorithm (such as {\em Ranking}) for the maximizing player. As such, for every bipartite graph $G$ it produces a distribution over matchings. For a bipartite graph $G\in G_n$, we use the following notation:

\begin{itemize}

\item $\rho_n(A,G)$ is the expected cardinality of matching produced by $A$ when the input graph is $G$.

\item $\rho_n(A,-)$ is the minimum over all $G \in G_n$ of $\rho_n(A,G)$. Namely, $\rho_n(A,-) = \min_{G\in G_n}[\rho_n(A,G)]$.

\item $\rho_n$ is the maximum over all $A$ (randomized online matching algorithms for the maximizing player) of $\rho_n(A,-)$. Namely, $\rho_n = \max_A[\rho_n(A,-)]$. (Showing that the maximum is attained is a technicality that we ignore here.)

\item $\rho = \inf_n \frac{\rho_n}{n}$. Namely, $\rho$ is the largest constant (independent of $n$) such that $\rho \cdot n \le \rho_n$ for all $n$. (One might find a definition such as $\rho = \lim_{n \rightarrow \infty} \frac{\rho_n}{n}$ more natural, but it turns out that both definitions of $\rho$ give the same value, which will be seen to be $1 - \frac{1}{e}$.)

\end{itemize}

Karp, Vazirani and Vazirani~\cite{KVV90} showed that $\rho_n(Ranking, -) \ge (1 - \frac{1}{e})n - o(n)$, where $e$ is the base of the natural logarithm (and $(1 - \frac{1}{e}) \simeq 0.632$). Unfortunately, that paper had only a conference version and not a journal version, and the proof presented in the conference version appears to have gaps. Later work (e.g., \cite{MSVV07,GM08,DJK13}), motivated by extensions of the online matching problem to other problems such as the {\em adwords} problem, presented alternative proofs, and also established that the $o(n)$ term is not required.
There have also been expositions of simpler versions of these proofs. See~\cite{BM08,Mathieu11,EFFS18}, for example. 
Summarizing this earlier work, we have:

\begin{theorem}
\label{thm:KVV90a}
For every bipartite graph $G\in G_n$, the expected cardinality of the matching produced by {\em Ranking} is at least $(1 - \frac{1}{e})n$. Hence $\rho_n(Ranking,-) \ge (1 - \frac{1}{e})n$, and $\rho \ge 1 - \frac{1}{e} \simeq 0.632$.
\end{theorem}

Karp, Vazirani and Vazirani~\cite{KVV90} also presented a distribution over $G_n$, and showed that for every online algorithm, the expected size of the matching produced (expectation taken over random choice of graph from this distribution) is at most $(1 - 1/e)n + o(n)$.
This distribution, that we shall refer to as $D_n$, is defined as follows. Select uniformly at random a permutation $\tau$ over $V$. For every $j$, the neighbors of vertex $u_j$ are $\{v_{\tau(j)}, \ldots, v_{\tau(n)}\}$. The unique perfect matching $M$ is the set of edges $(u_j, v_{\tau(j)})$ for $1 \le j \le n$.

To present the known results regarding $D_n$ more accurately, let as extend previous notation.

\begin{itemize}

\item $\rho_n(A,D_n)$ is the expected cardinality of matching produced by $A$ when the input graph $G$ is selected according to distribution $D_n$. (Hence expectation is taken both over randomness of $A$ and over selection from $D_n$.) By definition, for every algorithm $A$, $\rho_n(A,D_n)$ is an upper bound on $\rho_n(A,-)$.

\item $\rho_n(-,D_n)$ is the maximum over all $A$ (randomized online algorithms for the maximizing player) of $\rho_n(A,D_n)$. Namely, $\rho_n(-,D_n) = \max_A[\rho_n(A,D_n)]$. By definition, for every $n$, $\rho_n(-,D_n)$ is an upper bound on $\rho_n$.

\end{itemize}

It is not hard to see (and was shown also in Lemma 13 of~\cite{KVV90}) that for every two greedy online algorithms $A$ and $A'$ it holds that $\rho_n(A,D_n) = \rho_n(A',D_n)$. As greedy algorithms are optimal among online algorithms, and {\em Ranking} is a greedy algorithm, we have the following proposition.

\begin{proposition}
\label{pro:indifference}
For $D_n$ defined as above,
$$\rho_n(Ranking,D_n) = \rho_n(-,D_n) \ge \rho_n$$
\end{proposition}

The result of~\cite{KVV90} can be stated as showing that $\rho_n(-,D_n) \le (1 - \frac{1}{e})n + o(n)$. Later analysis (see for example~\cite{MSVV07}, or the lecture notes of Kleinberg~\cite{Kleinberg} or Karlin~\cite{Karlin}) replaced the $o(n)$ term by $O(1)$. Moreover, this upper bound holds not only for online randomized integral algorithms (that match edges as a whole), but also for online fractional algorithms (that match fractions of edges). Let us provide more details.

A fractional matching for a bipartite graph $G(U,V;E)$ is a nonnegative weight function $w$ for the edges such that for every vertex $u \in U$ we have $\sum_{v\in N(u)} w(u,v) \le 1$, and likewise, for every vertex $v\in V$ we have $\sum_{u\in N(v)} w(u,v) \le 1$. The size of a fractional matching is $\sum_{e \in E} w(e)$. It is well known (see~\cite{LP86}, for example) that in bipartite graphs, the size of the maximum fractional matching equals the cardinality of the maximum (integral) matching.

In the online bipartite fractional matching problem, as vertices of $U$ arrive, the maximizing player can add arbitrary positive weights to their incident edges, provided that the result remains a fractional matching. We extend the $\rho$ notation used for the integral case also to the fractional case, by adding a subscript $f$. Hence for example, $\rho_{f,n}(A,G)$ is the size of the fractional matching produced by an online algorithm $A$ when $G \in G_n$ is the input graph.

It is not hard to see that in the fractional setting, randomization does not help the maximizing player, in the sense that any randomized online algorithm $A$ for fractional matching can replaced by a deterministic algorithm $A'$ that on every input graph produces a fractional matching of at least the same size. (Upon arrival of vertex $u$, the fractional weight that $A'$ adds to edge $(u,v)$ equals the expected weight that $A$ adds to this edge, where expectation is taken over randomness of $A$.)  Consequently, $\rho_{f,n} \ge \rho_n$, and every upper bound on $\rho_{f,n}$ is also an upper bound on $\rho_n$.

The following theorem summarizes the known upper bounds on $\rho_{f,n}$, which are also the strongest known upper bounds on $\rho_n$.

\begin{theorem}
\label{thm:KVV90b}
For $D_n$ as defined above, $\rho_{f,n}(-,D_n) \le (1 - \frac{1}{e})n + O(1)$.
Consequently, $\rho_{n}(-,D_n) \le (1 - \frac{1}{e})n + O(1)$.
\end{theorem}

The combination of Theorems~\ref{thm:KVV90a} and~\ref{thm:KVV90b} implies the following corollary:

\begin{corollary}
\label{cor:KVV}
Using notation as above, $\rho = 1 - \frac{1}{e}$ and $\rho_f = 1 - \frac{1}{e}$.
The {\em Ranking} algorithm (which produces an integral matching) is asymptotically optimal (for the maximizing player) for online bipartite matching both in the integral and in the fractional case.
The distribution $D_n$ is asymptotically optimal for the minimizing player, both in the integral and in the fractional case.
\end{corollary}

In this manuscript, we shall be interested not only in the asymptotic ratios $\rho$ and $\rho_f$, but also in the exact ratios $\rho_n$ and $\rho_{f,n}$. Every (integral) matching is also a fractional matching, hence one may view {\em Ranking} also as an online algorithm for fractional matching.
As such, {\em Ranking} is easily seen not to be optimal for some $n$. For example, when $n = 4$, tedious but straighforward analysis shows that a different known algorithm referred to as {\em Balance} (see Section~\ref{sec:KVVb}) satisfies $\rho_{f,4}(Balance,-) > \rho_{f,4}(Ranking,-)$ (details omitted). However, for the integral case, it was conjectured in~\cite{KVV90} that both {\em Ranking} and $D_n$ are optimal for every $n$. Namely, the conjecture is:

\begin{conjecture}
\label{con:KVV}
$\rho_n = \rho_n(Ranking,D_n)$ for every $n$.
\end{conjecture}

The above conjecture, though still open, adds motivation (beyond Proposition~\ref{pro:indifference}) to determine the exact value of $\rho_n(Ranking,D_n)$.
This is done in the following theorem.

\begin{theorem}
\label{thm:exact}
Let the function $a(n)$ be such that $\rho_n(Ranking,D_n) = \frac{a(n)}{n!}$ for all $n$. Then $a(n) = (n+1)! - d(n+1) - d(n)$, where $d(n)$ is the number of derangements (permutations with no fixed points) on the numbers $[1,n]$. Consequently, $\rho_n(-,D_n) = (1 - \frac{1}{e})n + 1 - \frac{2}{e} + O(\frac{1}{n!}) \simeq (1 - \frac{1}{e})n  + 0.264$, and this is also an upper bound on $\rho_n$.
\end{theorem}

The rest of this paper is organized as follows.  In Section~\ref{sec:KVVb} we review a proof of Theorem~\ref{thm:KVV90b}. In doing so, we determine the value of the $O(1)$ term stated in the theorem, and also show that the upper bound is tight. Hence we end up proving the following theorem:

\begin{theorem}
\label{thm:fractionalopt}
For every $n$, {\em Balance} is the fractional online algorithm with best approximation ratio, $D_n$ is the distribution over graphs for which the approximation ratio is worst possible, and
$$\rho_{f,n} = \rho_{f,n}(Balance,D_n) = (1 - \frac{1}{e})n + \frac{1}{2} - \frac{1}{2e} + O(\frac{1}{n}) \simeq (1 - \frac{1}{e})n + 0.316$$
\end{theorem}

In Section~\ref{sec:exact} we prove Theorem~\ref{thm:exact}. The combination of Theorems~\ref{thm:exact} and~\ref{thm:fractionalopt} implies that $\rho_n < \rho_{f,n}$ for sufficiently large $n$. It also implies that $\rho_{f,n}(Balance,D_n) > \rho_{f,n}(Ranking,D_n)$ for sufficiently large $n$. Hence Proposition~\ref{pro:indifference} does not extend to online fractional matching.

In an appendix (Section~\ref{sec:KVVa}) we review a proof (due to~\cite{EFFS18}) of Theorem~\ref{thm:KVV90a}, and derive from it an upper bound of $(1 - \frac{1}{e})n + \frac{1}{e}$ on $\rho_n(Ranking,D_n)$. This last upper bound is weaker than the upper bounds of Theorems~\ref{thm:exact} and~\ref{thm:fractionalopt}, but its proof is different, and hence might turn out useful in attempts to resolve Conjecture~\ref{con:KVV}.

\subsection{Preliminaries -- {\em MonotoneG}}
\label{sec:MonotoneG}

When analyzing $\rho_n(Ranking,D_n)$ we shall use the following observation so as to simplify notation. Because {\em Ranking} is oblivious to names of vertices, the expected size of the matching produced by {\em Ranking} on every graph in the support of $D_n$ is the same. Hence we shall consider one representative graph from $D_n$, that we refer to as the monotone graph {\em MonotoneG}, in which $\gamma$ (in the definition of $D_n$) is the identity permutation. The monotone graph $G(U,V;E)$ satisfies $E = \{(u_i,v_j) \; | \; j \ge i\}$, and its unique perfect matching is $M = \{(u_i,v_i) \: | \; 1 \le i \le n \}$. Statements involving $\rho_n(Ranking,D_n)$ will be replaced by $\rho_n(Ranking,MonotoneG)$, as both expressions have the same value.

Likewise, the algorithm {\em Balance} is oblivious to names of vertices, and statements involving $\rho_{f,n}(Balance,D_n)$ will be replaced by $\rho_{f,n}(Balance,MonotoneG)$.



\section{Online fractional matchings}
\label{sec:KVVb}

Let us present a specific online fractional matching algorithm that is often referred to as {\em Balance}, which is the natural fractional analog of an algorithm by the same name introduced in~\cite{KP00}. {\em Balance} maintains a {\em load} $\ell(v)$ for every vertex $v \in V$, equal to the sum of weights of edges incident with $v$. Hence at all times, $0 \le \ell(v) \le 1$.
Upon arrival of a vertex $u$ with a set of neighbors $N(u)$, {\em Balance} distributes a weight of $\min[1, |N(u)| - \sum_{v\in N(u)} \ell(v)]$ among the edges incident with $u$, maintaining the resulting loads as balanced as possible. Namely, one computes a threshold $t$ such that $\sum_{v \in N(u) | \ell(v) < t} (t - \ell(v)) = \min[1, |N(u)| - \sum_{v\in N(u)} \ell(v)]$, and then adds fractional value $t - \ell(v)$ to each edge $(u,v)$ for those vertices $v \in N(u)$ that have load below $t$.

We first present a proof of Theorem~\ref{thm:KVV90b} based on previous work. The theorem is restated below, with the additive $O(1)$ term instantiated. Previous work either did not specify the $O(1)$ additive term (e.g., in~\cite{Karlin}), or derived an $O(1)$ term that is not tight (e.g., in~\cite{Kleinberg}).

\begin{theorem}
\label{thm:KVVb1}
For every $n$ it holds that
$$\rho_{f,n}(-,D_n) = (1 - \frac{1}{e})n + \frac{1}{2} - \frac{1}{2e} + O(\frac{1}{n}) \simeq (1 - \frac{1}{e})n + 0.316$$
Moreover, $\rho_{f,n}(-,D_n) = \rho_{f,n}(Balance,D_n)$.
\end{theorem}

\begin{proof}
For all graphs in the support of $D_n$, the size of the fractional matching produced by {\em Balance} is the same (by symmetry). Hence for simplicity of notation, consider the fractional matching produced by {\em Balance} when the input graph is the monotone graph {\em MonotoneG} (see Section~\ref{sec:MonotoneG}). It is not hard to see that when vertex $u_i$ arrives, {\em Balance} raises the load of each vertex in $\{v_i, \ldots, v_n\}$ by $\frac{1}{n-i+1}$. This can go on until the largest $k$ satisfying $\sum_{i=1}^k \frac{1}{n-i+1} \le 1$. Thereafter, when vertex $u_{k+1}$ arrives, {\em Balance} can raise the load of its $n-k$ neighbors from $\sum_{i=1}^k \frac{1}{n-i+1}$ to~1. Hence altogether the size of the fractional matching is precisely $k + (n-k)(1 - \sum_{i=1}^k \frac{1}{n-i+1})$, for $k$ as above.

The value of $k$ can be determined as follows. It is known that the harmonic number $H_n = \sum_{i=1}^n \frac{1}{i}$ satisfies $H_n = \ln n + \gamma + \frac{1}{2n} + O(\frac{1}{n^2})$, where $\gamma \simeq 0.577$ is the Euler-Mascheroni constant. $k$ is the largest integer such that $H_n - H_{n-k} \le 1$. Defining $\alpha \triangleq \frac{n-k}{n}$, we have that
$$H_n - H_{n-k} = \ln n + \gamma + \frac{1}{2n} + O(\frac{1}{n^2}) - \ln \alpha n - \gamma - \frac{1}{2\alpha n} + O(\frac{1}{n^2}) = \ln \frac{1}{\alpha} - \frac{\frac{1}{\alpha} - 1}{2n} + O(\frac{1}{n^2})$$
Choosing $\alpha = \frac{1}{e}$ (and temporarily ignoring the fact that in this case $k = (1 - \frac{1}{e})n$ is not an integer), we get that $H_n - H_{n-k} = 1  - \frac{e - 1}{2n} + O(\frac{1}{n^2})$. The size of a matching is then
$$(1 - \frac{1}{e})n + \frac{n}{e}(\frac{e - 1}{2n} + O(\frac{1}{n^2})) = (1 - \frac{1}{e})n + \frac{1}{2} + \frac{1}{2e} + O(\frac{1}{n})$$
as desired.

The fact that $k = (1 - \frac{1}{e})n$ above was not an integer requires that we round $k$ down to the nearest integer. The effect of this rounding is bounded by the effect of changing the number of neighbors available to $u_k$ and to $u_{k+1}$ by one (compared to the computation without the rounding). Given that the number of neighbors is roughly $\frac{n}{e}$, the overall effect on the size of the fractional matching is at most $O(\frac{1}{n})$.

We conclude that $\rho_{f,n}(Balance,D_n) = (1 - \frac{1}{e})n + \frac{1}{2} + \frac{1}{2e} + O(\frac{1}{n})$, implying that $\rho_{f,n}(-,D_n) \ge (1 - \frac{1}{e})n + \frac{1}{2} + \frac{1}{2e} + O(\frac{1}{n})$. In remains to show that $\rho_{f,n}(-,D_n) \le (1 - \frac{1}{e})n + \frac{1}{2} + \frac{1}{2e} + O(\frac{1}{n})$. This follows because {\em Balance} is the best possible online algorithm (for fractional bipartite matching) against $D_n$. Let us provide more details.

Given an input graph from the support of $D_n$, we shall say that a vertex $v \in V$ is {\em active} in round $i$ if it is a neighbor of $u_i$. Initially all vertices are active, and after every round, one more vertex (chosen at random among the active vertices) becomes inactive, and remains inactive forever. Let $a(i)$ denote the number of active vertices at the beginning of round $i$, and note that $a(i) = n - i + 1$. Consider an arbitrary online algorithm. Let $L(i)$ denote the average load of the active vertices at the beginning of round $i$. Then in round $i$, the average load first increases by at most $\frac{1}{a(i)}$ (as long as it does not exceed~1) by raising weights of edges, and thereafter, making one vertex inactive keeps the average load unchanged in expectation (over choice of input from $D_n$). Hence in expectation, in every round, the average load does not exceed the value of the average load obtained by {\em Balance}. This means that in every round, in expectation, the amount of unused load of the vertex that became inactive is smallest when the online maximizing algorithm is {\em Balance}. Summing over all rounds and using the linearity of expectation, {\em Balance} suffers the smallest sum of unused load, meaning that it maximizes the final expected sum (over all $V$) of loads. The sum of loads equals the size of the fractional matching.
\end{proof}

We now prove Theorem~\ref{thm:fractionalopt}.

\begin{proof}[Theorem~\ref{thm:fractionalopt}]
Given Theorem~\ref{thm:KVVb1}, it suffices to show that $\rho_{f,n}(Balance,-) = \rho_{f,n}(Balance,D_n)$, namely, that $D_n$ is the worst possible distribution over input graphs for the algorithm {\em Balance}. Moreover, given that {\em Balance} is oblivious to the names of vertices, it suffices to show that {\em MonotoneG} is the worst possible graph for {\em Balance}.

Let $G(U,V;E)\in G_n$ be a graph for which $\rho_{f,n}(Balance,G) = \rho_{f,n}(Balance,-)$.
As {\em Balance} is oblivious to the names of vertices, we may assume that $\{(u_i,v_i) | 1 \le i \le n \}$ is a perfect matching in $G$.

We use the notation $N(w)$ to denote the set of neighbors of a vertex $w$ in the graph $G$.
When running {\em Balance} on $G$, we use the notation $m(i,j)$ to denote the weight that the fractional matching places on edge $(u_i,v_j)$ (and $m(i,j) = 0$ if $(u_i,v_j) \not\in E$), and $m_i(j)$ to denote $\sum_{1 \le \ell \le i} m(u_{\ell},v_j)$. Clearly, $m_i(j)$ is non-decreasing in $i$.  The size of the final fractional matching is $m = \sum_{j=1}^n m_n(j)$. When referring to a graph $G'$, we shall use the notation $N'$ and $m'$ instead of $N$ and $m$.

An edge $(u_i,v_j)$ with $j < i$ is referred to as a {\em backward} edge.

\begin{proposition}
Without loss of generality, we may assume that $G$ has no backward edges. Hence $m_i(j) = m_j(j)$ for all $i > j$.
\end{proposition}

\begin{proof}
Suppose otherwise, and let $i$ be largest so that $u_i$ has backward edges. Modify $G$ by removing all backward edges incident with $u_i$, thus obtaining a graph $G'$. Compare the performance of {\em Balance} against the two graphs, $G$ and $G'$. On vertices $u_1, \ldots, u_{i-1}$, both graphs produce the same fractional matching. The extent to which $u_i$ is matched is at least as large in $G$ as it is in $G'$ (because also backward edges may participate in the fractional matching). Moreover, for every vertex $v_j$ for $i < j \le n$, it holds that $m'_i(j) \ge m_i(j)$. It follows that for every vertex $u_{\ell}$ for $\ell > i$, its marginal contribution to the fractional matching in $G$ is at least as large as its marginal contribution  in $G'$. Hence the fractional matching produced by {\em Balance} for $G'$ is not larger than that produced for $G$. Repeating the above argument, all backward edges can be eliminated from $G$ without increasing the size of the fractional matching.
\end{proof}

\begin{lemma}
\label{lem:mMonotone}
Without loss of generality we may assume that:

\begin{enumerate}
\item $m_i(i) \le m_j(j)$ (or equivalently, $m_n(i) \le m_n(j)$) for all $i < j$.
\item $m_i(i) \ge m_i(j)$ for all $i$ and $j$.
\end{enumerate}
\end{lemma}

\begin{proof}
We first present some useful observations. For $1 \le i < n$, consider the set $N(u_i)$ of neighbours  of $u_i$ in $G$ (and recall that $v_i \in N(u_i)$, and that there are no backward edges). Then without loss of generality we may assume that $m_i(i) \ge m_i(j)$ for all $v_j \in N(u_i)$.  This is because if there is some vertex $v_j \in N(u_i)$ with $m_i(j) > m_i(i)$, then
it must hold (by the properties of {\em Balance}) that $m(i,j)=0$.  Hence
the run of {\em Balance} would not change if the edge $(u_i,v_j)$ is removed from $G$ (and then $v_j \not\in N(u_i)$).

Moreover, we may assume that $m_i(i) = m_i(j)$ for all $v_j \in N(u_i)$. Suppose otherwise. Then for $v_j \in N(u_i)$ with smallest $m_i(j)$, modify $G$ to a graph $G'$ as follows. For all $\ell < i$, make $u_{\ell}$ a neighbor of $v_i$ iff it was a neighbor of $v_j$, and make $u_{\ell}$ a neighbor of $v_j$ iff it was a neighbor of $v_i$.
The final size of the fractional matching in $G'$ (which is $\sum_{j=1}^n m'_n(j)$) cannot be larger than in $G$. This is because $m'_i(i) < m_i(i)$, $m'_i(j) > m_i(j)$ and for $\ell \not= j$ satisfying $\ell > i$ it holds that $m'_i(\ell) = m_i(\ell)$. Moreover, as $m_i(i) < m_i(j) \le 1$, $u_i$ is fully matched in $G$ and hence also in $G'$, so the total size of fractional matching after step $i$ is the same in both graphs. Thereafter, the marginal increase of the fractional matching at each step cannot be larger in $G'$ than it is in $G$.

By the same arguments as above we may assume that $m_{i+1}(i+1) = m_{i+i}(j)$ for all $v_j \in N(u_{i+1})$.

Suppose now that item~1 fails to hold. Then for some $1 \le i \le n-1$ it holds that $m_i(i) > m_{i+1}(i+1)$.
Vertices $u_i$ and $u_{i+1}$ cannot have a common neighbor because if they do (say, $v_{\ell}$) it holds that  $m_{i+1}(i+1) = m_{i+1}(\ell) \ge m_i(\ell) = m_i(i_i)$. Hence we may exchange the order of $u_i$ and $u_{i+1}$ (and likewise $v_i$ and $v_{i+1}$) without affecting the size of the fractional matching produced by {\em Balance}.

Repeating the above argument whenever needed we prove item~1 of the lemma.

For $j < i$ item~2 holds because $m_i(j) = m_j(j) \le m_i(i)$ (the last inequality follows from item~1). For $j > i$ item 2 holds because at the first point in time $\ell \le i$ in which $m_{\ell}(j) = m_i(j)$ it must be that $m_{\ell}(j) = m_\ell(\ell)$, and item~1 implies that $m_{\ell}(\ell) \le m_i(i)$.
\end{proof}

It is useful to note that Lemma~\ref{lem:mMonotone} implies that there is some round number $t$ such that for all $\ell \ge t$ vertex $v_{\ell}$ is fully matched (namely, $m_n(\ell) = 1$), and for every $\ell < t$ vertex $v_{\ell}$ is not fully matched (namely, $m_n(\ell) < 1$). As to vertices in $u$, for $\ell < t$ vertex $u_{\ell}$ is fully matched, for $\ell > t$ vertex $u_{\ell}$ contributes nothing to the fractional matching, and $u_t$ is either partly matched or fully matched. Recalling that $m$ denotes the size of the final fractional matching, we thus have (for $t$ as above):

\begin{equation}
\label{eq:m}
m = t- 1 + \sum_{j \ge t} m(t,j)
\end{equation}

At every step $i$, the contribution of vertex $v_i$ towards the fractional matching is finalized at that step, namely, $m_n(i) = m_i(i)$.
Lemma~\ref{lem:mMonotone} implies that for the worst graph $G$, this vertex $v_i$ is the one with largest $m_i$ value at this given step. Hence $m_i(i) =  \max_{j \ge i}[m_i(j)]$ and we have:
$$m = \sum_{i=1}^n m_n(i) = \sum_{i=1}^n m_i(i) = \sum_{i=1}^n \max_{j \ge i}[m_i(j)].$$
At this point it is intuitively clear why {\em MonotoneG} is the graph in $G_n$ on which {\em Balance} produces the smallest fractional matching. This is because with {\em MonotoneG}, at each step $i$ the fractional matching gets credited a value $m_i(i)$ that is the average of the values $m_i(j)$ for $j \ge i$, whereas for $G$ its gets credited the maximum of these values. Below we make this argument rigorous.

Consider an alternative {\em averaging process} replacing algorithm {\em Balance}. It uses the same fractional matching as in {\em Balance} and the same $m(i,j)$ values, but maintains values $m'_i(i)$ that may differ from $m_i(i)$. 
At round 1, instead of being credited the maximum $m_1(1) = \max_{j \ge 1}[m_1(j)]$, the process is credited only the average $m'_i(1) = \frac{1}{n}\sum_{j=1}^n m_i(j)$. The remaining $\max_{j \ge 1}[m_1(j)] - \frac{1}{n}\sum_{j=1}^n m_1(j)$ is referred to as the {\em slackness} $s(1)$. More generally, at every round $i > 1$, instead of being credited by $\max_{j \ge i}[m_i(j)]$ at step~$i$, the averaging process gets credit from two sources. One part of the credit is the average $\frac{1}{n-i+1}\sum_{j=i}^n m_i(j)$,  where $s(i) = \max_{j \ge i}[m_i(j)] - \frac{1}{n-i+1}\sum_{j=i}^n m_i(j)$ is the slackness generated at round $i$. In addition, the process gets credit also for the slackness accumulated in previous rounds $\ell < i$, in such a way that each slackness variable $s(\ell)$ gets distributed evenly among the $n - \ell$ rounds that follow it. Hence we set
\begin{equation}
\label{eq:defm'}
m'_i(i) = \frac{1}{n-i+1}\sum_{j=i}^n m_i(j) + \sum_{\ell=1}^{i-1} \frac{s(\ell)}{n-\ell}.
\end{equation}
The averaging process continues until the first round $t'$ at which $m'_{t'}(t') \ge 1$, at which point $m'_j(j)$ is set to~1 for all $j \ge t'$, and the process ends. The size of the fractional matching associated with the averaging process is $m' = \sum_{i=1}^n m'_i(i)$. Computing $m'$ using the contributions of the vertices from $U$, for $t'$ as above, we get that:

\begin{equation}
\label{eq:m'}
m' = t' - 1 + \sum_{j \ge t'} m(t',j)
\end{equation}


\begin{proposition}
For the graph $G$, the size of the fractional matching produced by the averaging process is no larger than that produced by {\em Balance}. Namely, $m' \le m$.
\end{proposition}

\begin{proof}
Compare Equations (\ref{eq:m}) and (\ref{eq:m'}). If $t' = t$ then $m' = m$, and if $t' < t$ then $m' < m$. Hence it suffices to show that the assumption $t' \ge t$ implies that $t' = t$. This follows because $m_t(j) = 1$ for all $j \le t$ (as noted above), and so:
$$m'_t(t) = \frac{1}{n-t+1}\sum_{j=t}^n m_t(j) + \sum_{\ell=1}^{t-1} \frac{s(\ell)}{n-\ell} = 1 + \sum_{\ell=1}^{t-1} \frac{s(\ell)}{n-\ell} \ge 1$$
where the last inequality holds because all slackness variables $s(\ell)$ are non-negative.
\end{proof}

\begin{proposition}
For {\em MonotoneG}, running the averaging process and running {\em Balance} are exactly the same process, giving $m'(MonotoneG) = m(MonotoneG)$.
\end{proposition}

\begin{proof}
This is because when running {\em Balance} on {\em MonotoneG}, at every round $i$ we have that $m_i(i) = m_i(j)$ for all $j > i$. Hence there is no difference between the average and the maximum of the $m_i(j)$ for $j \ge i$.
\end{proof}

\begin{proposition}
The size of the fractional matching produced by the averaging process for graph $G$ is not smaller than the size it produces for {\em MonotoneG}. Namely, $m'(G) \ge m'(MonotoneG)$.
\end{proposition}

\begin{proof}
Running the averaging process on graph $G$, we claim that for every round $i < t'$ we have that:

\begin{equation}
\label{eq:3}
m'_i(i) = \sum_{k\le i} \frac{1}{n-k+1}
\end{equation}

The equality can be proved by induction. For $i=1$ both sides of the equality are~$\frac{1}{n}$. For the inductive step, recalling Equation~\ref{eq:defm'} one can infer that
$$m'_{i+1}(i+1) = \frac{1}{n-i}\left((n - i + 1)m'_i(i) - m'_i(i) + 1\right)$$
where the $+1$ term is because $i < t'$. Likewise, the right hand side develops in the same way:
$$\sum_{k\le i+1} \frac{1}{n-k+1} = \frac{1}{n-i}\left((n - i + 1)\sum_{k\le i} \frac{1}{n-k+1} - \sum_{k\le i} \frac{1}{n-k+1} + 1\right)$$

The left hand side of Equation~(\ref{eq:3}) concerns graph $G$. Observe that $m'_i(i)$ for {\em MonotoneG} exactly equals the right hand side of Equation~(\ref{eq:3}). It follows that the averaging process ends at the same step $t'$ both on the graph $G$ and on {\em MonotoneG}, and up to step $t'$ the accumulated fractional matching $m'$ is identical. For rounds $j \ge t'$ we have that $m'_{j}(j) = 1$ for $G$ and it cannot be larger than~1 for {\em MonotoneG}, proving the proposition.
\end{proof}

Combining the three propositions above we get that:

$$m(G) \ge m'(G) \ge m'(MonotoneG) = m(MonotoneG)$$

This completes the proof of Theorem~\ref{thm:fractionalopt}. \end{proof}

\section{Online integral matching}
\label{sec:exact}

The first part of Theorem~\ref{thm:exact} is restated in the following theorem (recall the definition of the monotone graph {\em MonotoneG} in Section~\ref{sec:MonotoneG}).

\begin{theorem}
\label{thm:exact1}
Let the function $a(n)$ be such that $\rho_n(Ranking,MonotoneG) = \frac{a(n)}{n!}$ for all $n$. Then $a(n) = (n+1)! - d(n+1) - d(n)$, where $d(n)$ is the number of derangements (permutations with no fixed points) on the numbers $[1,n]$.
\end{theorem}

\begin{proof}
When the input is {\em MonotoneG}, then for every permutation $\pi$ used by {\em Ranking}, the matching $M'$ produced satisfies the following two properties:

\begin{itemize}

\item All vertices in some prefix of $U$ are matched, and then no vertices in the resulting suffix are matched. 
    This is because all neighbors of $u_{j+1}$ are also neighbors of $u_j$, so if $u_{j+1}$ is matched then so is $u_j$.

\item The order in which vertices of $V$ are matched is consistent with the order $\pi$ (for those vertices that are matched -- some vertices of $V$ may remain unmatched). In other words, if two vertices $v_i$ and $v_j$ are matched and $\pi(i) < \pi(j)$, then the vertex $u \in U$ matched with $v_i$ arrived earlier (has smaller index) than the vertex $u' \in U$ matched with $v_j$.

\end{itemize}

Some arguments in the proof that follows make use of the above properties, without explicitly referring to them.

Fix $n$ and {\em MonotoneG} as input. Let $\Pi_n$ denote the set of all permutations over $V$. Hence $|\Pi_n| = n!$. {\em Ranking} picks one permutation $\pi \in \Pi_n$ uniformly at random. Recall our notation that $\pi(i)$ is the rank of $v_i$ under $\pi$.
We shall use $\pi_i$  to denote the item of rank $i$ in $\pi$ (namely, $\pi_i = \pi^{-1}(i)$).
For $i \le n$, let $a(n,i)$ denote the number of permutations $\pi\in \Pi_n$ under which $\pi_i$ is matched.

\begin{proposition}
\label{pro:a(n)}
For $a(n)$ as defined in Theorem~\ref{thm:exact1} and $a(n,i)$ as defined above, it holds that $a(n) = \sum_{i=1}^n a(n,i)$.
\end{proposition}

\begin{proof}
For a permutation $\pi \in \Pi_n$, let $x(\pi)$ denote the size of the greedy matching produced when {\em Ranking} uses $\pi$ and the input graph in {\em MonotoneG}. Then by definition:
$$a(n) = \sum_{\pi \in \Pi_n} x(\pi).$$
By changing the order of summation:
$$\sum_{\pi \in \Pi_n} x(\pi) = \sum_{i=1}^n a(n,i).$$
Combining the above equalities 
proves the proposition.
\end{proof}

Proposition~\ref{pro:a(n)} motivates the study of the function $a(n,i)$.

\begin{lemma}
\label{lem:a(n,i)}
The function $a(n,i)$ satisfies the following:

\begin{enumerate}
\item $a(n,1) = n!$ for every $n \ge 1$.
\item $a(n,i) = a(n,i+1) + a(n-1,i)$ for every $1 \le i < n$.
\end{enumerate}
\end{lemma}

\begin{proof}
The first statement in the lemma holds because in every permutation $\pi$, the item $\pi_1$ is matched with $u_1$. Hence it remains to prove the second statement.

Fixing $n > 1$ and $i < n$, consider the following bijection $B_i : \Pi_n \longrightarrow \Pi_n$, where given a permutation $\pi \in \Pi_n$, $B_i(\pi)$ flips the order between locations $i$ and $i+1$. Namely, $B_i(\pi)_i = \pi_{i+1}$ and $B_i(\pi)_{i+1} = \pi_i$ (we use $B_i(\pi)_i$ as shorthand notation for $(B_i(\pi))_i$). We compare the events that $\pi_i$ is matched by the greedy matching when {\em Ranking} uses $\pi$ with the event that $B_i(\pi)_{i+1}$ is matched by the greedy matching when {\em Ranking} uses $B_i(\pi)$.

There are four possible events:

\begin{enumerate}

\item Both $\pi_i$ and $B_i(\pi)_{i+1}$ are matched.

\item Neither $\pi_i$ nor $B_i(\pi)_{i+1}$ are matched.

\item $\pi_i$ is matched but $B_i(\pi)_{i+1}$ is not matched.

\item $\pi_i$ is not matched but $B_i(\pi)_{i+1}$ is matched.

\end{enumerate}

Though any of the first three events may happen, the fourth event cannot possibly happen. This is because the item in location $i+1$ in $B_i(\pi)$ is moved forward to location $i$ in $\pi$, so if the greedy algorithm matches it (say to $u_j$) in $B_i(\pi)$, then the greedy algorithm must match it (either to the same $u_j$ or to the earlier $u_{j-1}$) in $\pi$.

It follows that $a(n,i) - a(n,i+1)$ exactly equals the number of permutations in which the third event happens. Hence we characterize the conditions under which the third event happens. Let $u_j$ be the vertex matched with $\pi_i$ in $\pi$.
Up to the arrival of $u_j$, the behavior of {\em Ranking} on $B_i(\pi)$ and $\pi$ is identical. Thereafter, for $u_j$ not to be matched to $B_i(\pi)_{i+1} = \pi_i$, it must be matched to the earlier $B_i(\pi)_i$. Thereafter, for $u_{j+1}$ not to be matched to $B_i(\pi)_{i+1}$, it must be that $B_i(\pi)_{i+1}$ is not a neighbor of $u_{j+1}$. But $B_i(\pi)_{i+1} = \pi_i$ is a neighbor of $u_j$ (it was matched to $u_j$ under $\pi$), and hence it must be that $\pi_i = v_j$. Summarizing, the third event happens if and only if the permutation $B_i(\pi)$ comes from the following class $\hat{\Pi}$, where permutations $\hat{\pi}\in \hat{\Pi}$ are those that have the property that $\hat{\pi}_i$ is matched, and $\hat{\pi}_{i+1} = v_j$, for the same $j$ for which $u_j$ is the vertex matched with $\hat{\pi}_i$. Consequently, $a(n,i) = a(n,i+1) + |\hat{\Pi}|$.

To complete the proof of the lemma, it remains to show that $|\hat{\Pi}| = a(n-1,i)$. Let $\Pi' \subset |\Pi_{n-1}|$ be the set of these permutations $\pi' \in \Pi_{n-1}$ under which {\em Ranking} (when $|U| = |V| = n-1$) matches the item $\pi'_{i}$.

\begin{claim}
For $\hat{\Pi}$ and $\Pi'$ as defined above it holds that $|\hat{\Pi}|  = |\Pi'|$.
\end{claim}

\begin{proof}
We first show a mapping from $\hat{\Pi}$ to $\Pi'$. Given $\hat{\pi} \in \hat{\Pi}$, let $v_j = \hat{\pi}_{i+1}$. To obtain permutation  $\pi' \in \Pi_{n-1}$ from $\hat{\pi}$, remove $v_j$ from $\hat{\pi}$, identify location $k$ in $\hat{\pi}$ with location $k-1$ in $\pi'$ (for $i+2 \le k \le n$), and identify item $v_{\ell}$ of $\hat{\pi}$ with item $v_{\ell-1}$ of $\pi'$ (for $j+1 \le \ell \le n$). We show now that $\pi' \in \Pi'$ (namely, $\pi'_i$ is matched, when the input graph is {\em MonotoneG} with $|U| = |V| = n-1$).

The vertices $u_1, \ldots, u_{j-1}$ are matched to exactly the same locations in $\pi'$ and in $\hat{\pi}$, because the only vertices whose indices were decremented had index $\ell \ge j+1$, and are neighbors of $u_1, \ldots, u_{j-1}$ both before and after the decrement. Let $v_k = \hat{\pi}_i$  and note that $k > j$, because $v_k$ is matched to $u_j$ and it is not $v_j = \hat{\pi}_{i+1}$.  Hence $\pi'_i = v_{k-1}$ and it too is a neighbor of $u_j$, because $j \le k-1$. Hence $\pi'_i$ will be matched to $u_j$.

Conversely, we have the following mapping from $\Pi'$ to $\hat{\Pi}$. Given $\pi' \in \Pi'$, let $u_j$ be the vertex matched $\pi'_i$. To obtain permutation  $\hat{\pi} \in \hat{\Pi}$ from $\pi'$, identify location $k$ in $\hat{\pi}$ with location $k-1$ in $\pi'$ (for $i+2 \le k \le n$), identify item $v_{\ell}$ of $\hat{\pi}$ with item $v_{\ell-1}$ of $\pi'$ (for $j+1 \le \ell \le n$), and set $\hat{\pi}_{i+1} = v_{j}$. We show now that $\hat{\pi} \in \hat{\Pi}$.

As in the first mapping, the vertices $u_1, \ldots, u_{j-1}$ are matched to exactly the same locations in $\pi'$ and in $\pi$. Let $v_k = \pi'_i$ and note that $k \ge j$, because $v_k$ was matched to $u_j$.  Hence $\hat{\pi}_i = v_{k+1}$ is neighbor of $u_j$, and will be matched to $u_j$. On the other hand, $\hat{\pi}_{i+1} = v_j$ will not be matched because it is not a neighbor of any of $[u_{j+1}, u_n]$. Hence $\hat{\pi} \in \hat{\Pi}$.

Given the two mappings described above (one is the inverse of the other) we have a bijection between $\Pi'$ and $\hat{\Pi}$, proving the claim.
\end{proof}

The claim above implies that $|\hat{\Pi}| = |\Pi'| = a(n-1,i)$, and consequently that $a(n,i) = a(n,i+1) + a(n-1,i)$, proving the lemma.
\end{proof}

In passing, we note the following corollary.

\begin{corollary}
\label{cor:a(n)}
For $a(n,i)$ and $a(n)$ as defined above, $a(n) = (n+1)! - a(n+1,n+1)$.
\end{corollary}

\begin{proof}
Using item~1 of Lemma~\ref{lem:a(n,i)} we have that $a(n+1,1) = (n+1)!$. Applying item~2 of Lemma~\ref{lem:a(n,i)} iteratively for all $1 \le i \le n$ we have that $a(n+1,1) - a(n+1,n+1) = \sum_{i=1}^n a(n,i)$. Proposition~\ref{pro:a(n)} shows that $\sum_{i=1}^n a(n,i) = a(n)$. Combining these three equalities we obtain $a(n) = (n+1)! - a(n+1,n+1)$, as desired.
\end{proof}

Corollary~\ref{cor:a(n)} can also be proved directly, without reference to Lemma~\ref{lem:a(n,i)}. See Appendix~\ref{sec:directly} for details.

To obtain expressions for the values $a(n,i)$, let us introduce additional notation. A {\em fixpoint} (or {\em fixed point}) in a permutation $\pi$ is an item that does not change its location under $\pi$ (namely, $\pi(i) = i$). For $n \ge 1$ and $1 \le i \le n$ define $d(n,i)$ be the number of permutations over $[n]$ in which the only fixpoints (if any) are among the first $i$ items. For example, $d(3,1) = 3$ due to the permutations 132 (only 1 is a fixed point) 231 (no fixpoints) and 312 (no fixpoints).

\begin{lemma}
\label{lem:d(n,i)}
The function $d(n,i)$ satisfies the following:

\begin{enumerate}
\item $d(n,n) = n!$ for every $n \ge 1$.
\item $d(n,i+1) = d(n,i) + d(n-1,i)$ for every $1 \le i < n$.
\end{enumerate}
\end{lemma}

\begin{proof}
$d(n,n)$ denotes the number of permutations on $[n]$ with no restrictions, and hence $d(n,n) = n!$, which is the first statement of the lemma.

Consider now the second statement of the lemma. Let $\Pi_{n,i}$ denote the set of permutations in which the only fixpoints (if any) are among the first $i$ items. Then the second statement asserts that $|\Pi_{n,i+1}| = |\Pi_{n,i}| + |\Pi_{n-1,i}|$. The set $\Pi_{n,i+1}$ can be partitioned in two. In one part $i+1$ is not a fixpoint. This part is precisely $\Pi_{n,i}$. In the second part, $i+1$ is a fixpoint. To specify a permutation in this part we need to specify the location of the remaining $n-1$ items, where the only fixpoints allowed are among the first $i$ items. The number of permutations satisfying these constraints is  $\Pi_{n-1,i}$, by definition. Hence indeed $|\Pi_{n,i+1}| = |\Pi_{n,i}| + |\Pi_{n-1,i}|$, proving the lemma.
\end{proof}

\begin{corollary}
\label{cor:ad}
For every $n \ge 1$ and $1 \le i \le n$ it holds that $a(n,i) = d(n,n+1-i)$.
\end{corollary}

\begin{proof}
The proof is by induction on $n$, and for every value of $n$, by induction on $i$.

For the base case $n = 1$, necessarily $i=1$ (and hence also $n + 1 - i =1$) and indeed we have $a(1,1) = 1 = d(1,1)$. Fixing $n > 1$, the base case for $i$ is $i = 1$ (and $n+1-1 = n$) and indeed we have that $a(n,1) = n! = d(n,n)$. For the inductive step, consider $a(n,i)$ with $n > 1$ and $1 < i \le n$, and assume the inductive hypothesis for $n' < n$ and the inductive hypothesis for $n$ and $i' < i$. Then we have:

$$a(n,i) = a(n,i-1) - a(n-1,i-1) = d(n, n - i + 2) - d(n-1, n - i + 1) = d(n, n-i+1)$$

The first equality is by Lemma~\ref{lem:a(n,i)}, the second equality is by the inductive hypothesis, and the third equality is by Lemma~\ref{lem:d(n,i)}.
\end{proof}




We can now complete the proof of Theorem~\ref{thm:exact1}. By Corollary~\ref{cor:a(n)} we have that $a(n) = (n+1)! - a(n+1,n+1)$. By Corollary~\ref{cor:ad} we have that $a(n+1,n+1) = d(n+1,1)$. By definition, $d(n+1,1)$ is the number of permutations on $[n+1]$ in which only item~1 is allowed to be a fixpoint. This number is precisely $d(n+1) + d(n)$ (where $d(j)$ are the derangement numbers), where the term $d(n+1)$ counts those permutations in which there is no fixpoint, and the term $d(n)$ counts those permutations in which item~1 is the only fixpoint.
\end{proof}

The second part of Theorem~\ref{thm:exact} is restated in the following Corollary.

\begin{corollary}
\label{cor:exact}
For every $n$,
$$\rho_n(Ranking,MonotoneG) = (1 + \frac{1}{e})n + (1 - \frac{2}{e}) + \nu(n)$$
where $|\nu(n)| < \frac{1}{n!}$.
\end{corollary}

\begin{proof}
Theorem~\ref{thm:exact1} shows that $a(n) = (n+1)! - d(n+1) - d(n)$, where $d(n)$ are the derangement numbers. It is known that $d(n) = \frac{n!}{e}$ rounded to the nearest integer. Hence $|d(n) - \frac{n!}{e}| < \frac{1}{2}$ and $|d(n+1) + d(n) - \frac{(n+1)!}{e} - \frac{n!}{e}| < 1$. Hence $|a(n) - (1 - \frac{1}{e})(n+1)! - \frac{n!}{e}| < 1$. Dividing by $n!$ and replacing $(1 - \frac{1}{e})(n+1)$ by $(1 - \frac{1}{e})n + 1 - \frac{1}{e}$ the corollary is proved.
\end{proof}

\subsection{Some related sequences}

To illustrate the values of some of the parameters involved in the proof of Theorem~\ref{thm:exact1}, consider a triangular table $T$ where row $n$ has $n$ columns. The entries (for $1 \le i \le n$) are $d(n,i)$, as defined prior to Lemma~\ref{lem:d(n,i)}. Recall that $d(n,i) = a(n, n+1-i)$, hence the table also provides the $a(n,i)$ values. We initialize the diagonal of the table by $d(n,n) = n!$. Thereafter we fill the remaining cells of table row by row, by using the relation $d(n,i) = d(n,i+1) - d(n-1,i)$, implied by Lemma~\ref{lem:d(n,i)}. Finally, compute $a(n) = \sum_{i=1}^n a(n,i) = \sum_{i=1}^n d(n,i)$ by summing up each row. The table below shows the computation of $a(n)$ for $n \le 6$.
\medskip

\begin{tabular}{|c|c|c|c|c|c|c|c|}
  \hline
  n & d(n,1)=a(n,n) & d(n,2) & d(n,3) & d(n,4) & d(n,5) & d(n,6) & a(n) \\
  \hline
  1 & 1 &  &  &  &  &  & 1 \\
  2 & 1 & 2 &  &  &  &  & 3 \\
  3 & 3 & 4 & 6 &  &  &  & 13 \\
  4 & 11 & 14 & 18 & 24 &  &  & 67 \\
  5 & 53 & 64 & 78 & 96 & 120 &  & 411 \\
  6 & 309 & 362 & 426 & 504 & 600 & 720 & 2921 \\
  \hline
\end{tabular}
\medskip

The table $T$ is identical in its definition to Sequence A116853, named {\em Difference triangle of factorial numbers read by upward diagonals}, in {\em The Online Encyclopedia of Integer Sequences}~\cite{Sloane}. The row sums (and hence $a(n)$) in this table give Sequence A180191 (with an offset of~1 in the value of $n$), named {\em Number of permutations of $[n]$ having at least one succession}. The first column (which equals $a(n,n)$) is the sequence A000255. These relations between $a(n)$ and the various sequences in~\cite{Sloane} helped guide the statement and proof of Theorem~\ref{thm:exact1}.

The derangement numbers $d(n)$ (which form the sequence A000166) can be easily computed by the recurrence $d(n) = n\cdot d(n-1) + (-1)^n$ (due to Euler). The table below shows the computation of $a(n) = (n+1)! - d(n+1) - d(n)$ for $n \le 7$.
\medskip

\begin{tabular}{|c|c|c|c|}
  \hline
  n & n! & d(n) & a(n) \\
  \hline
  1 & 1 & 0 & 1 \\
  2 & 2 & 1 & 3 \\
  3 & 6 & 2 & 13 \\
  4 & 24 & 9 & 67 \\
  5 & 120 & 44 & 411 \\
  6 & 720 & 265 & 2921 \\
  7 & 5040 & 1854 & 23633 \\
  8 & 40320 & 14833 &  \\
  \hline
\end{tabular}
\medskip

\subsection*{Acknowledgements}

The work of the author is supported in part by the Israel Science Foundation (grant No. 1388/16). The results reported in this manuscript were obtained in preparation for a talk given at the event {\em Building Bridges II: Conference to celebrate the 70th birthday of Laszlo Lovasz, Budapest, July 2018}.
I thank several people whose input helped shape this work.
Alon Eden and Michal Feldman directed me to the proof presented in~\cite{EFFS18}, which is the one presented here (in the appendix) for Theorem~\ref{thm:KVV90a}.
Thomas Kesselheim and Aranyak Mehta directed me to additional relevant references.
The statement and proof of Theorem~\ref{thm:exact1} were based on noting some numerical coincidences between the values of $a(n)$ for small $n$ and sequences in {\em The Online Encyclopedia of Integer Sequences}~\cite{Sloane}. Dror Feige wrote a computer program that computes $a(n)$, which made these numerical coincidences evident.  Alois Heinz offered useful advice as to how to figure out proofs for various identities claimed in~\cite{Sloane}.

\begin{appendix}

\section{A performance guarantee for {\em Ranking}}
\label{sec:KVVa}

For completeness, we review here a proof of Theorem~\ref{thm:KVV90a}.
The proof that we present uses essentially the same mathematical expressions as the proof presented in~\cite{DJK13}. A simple presentation of the proof of~\cite{DJK13} appeared in a blog post of Claire Mathieu~\cite{Mathieu11} (with further slight simplifications made possible by a comment provided there by Pushkar Tripathi). We shall give an arguably even simpler presentation, due to Eden, Feldman, Fiat and Segal~\cite{EFFS18}. The proofs in~\cite{DJK13,Mathieu11} make use of linear programming duality. The proof below is based on an economic interpretation, and a proof technique
that splits {\em welfare} into the sum of {\em utility} and {\em revenue}. These last two terms turn out to be scaled versions of the dual variables used in~\cite{DJK13,Mathieu11}, but the proof does not need to make use of LP duality.

\begin{proof}{\bf [Theorem~\ref{thm:KVV90a}]}
Fix an arbitrary perfect matching $M$ in $G$. Given a vertex $v \in V$, we use $M(v)$ to denote the vertex in $U$ matched with $v$ under $M$.

Recall that {\em Ranking} chooses a random permutation $\pi$ over $V$. Equivalently, we may assume that every vertex $v_i \in V$  chooses independently uniformly at random a real valued weight $w_i \in [0,1]$, and then the vertices of $V$ are sorted in order of increasing weight (lowest weight first). This gives a random permutation $\pi$. The same permutation $\pi$ is also obtained if each weight $w_i$ is replaced by a ``price" $p_i = e^{w_i - 1}$ and vertices are sorted by prices (because $e^{x-1}$ is a monotonically increasing function in $x$). Observe that $p_i \in [\frac{1}{e},1]$, though it is not uniformly distributed in that range. The expected price that {\em Ranking} assigns to an item is:
\begin{equation}
\label{eq:price}
E[p_i] =  \int_0^1 e^{w_i - 1} dw_i = \frac{1}{e}(e - 1) = 1 - \frac{1}{e}
\end{equation}

It is convenient to think of the vertices of $U$ as {\em buyers} and the vertices of $V$ as {\em items}. Suppose that given $G(U,V;E)$, each vertex (buyer) $u \in U$ desires only items $v \in V$ that are neighbors of $u$ (namely, $u$ desires $v$ iff $(u,v) \in E$), is willing to pay~1 for any such item, and wishes to buy exactly one item. The seller holding the items is offering to sell each item $v_i$ for a price of $p_i$. Then given $G$, the matching produced by executing the {\em Ranking} algorithm is the same as the one that would be produced in a setting in which each buyer $u_j$, upon arrival, buys its cheapest exposed desired item, if there is any. If $p_i$ is the price of the purchased item $v_i$, then the {\em revenue} that the seller extracts from the sale of $v_i$ to $u_j$ is $r(v_i) = p_i$, whereas the {\em utility} that the buyer extracts is $y(u_i) = 1 - p_i$. Consequently, the revenue plus utility extracted from a sale is~1, and the total revenue extracted from all sales plus the total utility sum up to exactly the cardinality of the matching.

To lower bound the expected cardinality of the matching, we consider each edge $(M(v_i),v_i) \in M$ separately, and consider the expectation $E[r(v_i) + y(M(v_i))]$, where expectation is taken over the choice of $\pi$. Using the linearity of the expectation, we will have that $\rho_n(Ranking,G) = \sum_{v_i \in V} E[r(v_i) + y(M(v_i))]$.

\begin{lemma}
\label{lem:claims}
For every $v_i \in V$ it holds that $E[r(v_i) + y(M(v_i))] \ge 1 - \frac{1}{e}$. Moreover, this holds even if expectation is taken only over the choice of random weight $w_i$ (and hence of random price $p_i$) of item $v_i$, without need to consider other aspects of the random permutation $\pi$.
\end{lemma}

\begin{proof}
Fix an arbitrary graph $G(U,V;E)\in G_n$, 
an arbitrary perfect matching $M$, and arbitrary prices $p_j \in [\frac{1}{e},1]$ for all items $v_j \not= v_i$. The price $p_i$ for item $v_i$ is set at random. Let $M'$ denote the greedy matching produced by this realization of the {\em Ranking} algorithm (where each buyer upon its arrival is matched to the exposed vertex of lowest price among its neighbors, if there is any). Suppose as a thought experiment that item $v_i$ is removed from $V$, and consider the greedy matching $M'_{-i}$ that would have been produced in this setting. Let $p$ denote the price of the item in $V$ matched to $M(v_i)$ under $M'_{-i}$, and set $p=1$ if $M(v_i)$ is left unmatched under $M'_{-i}$. Now we make two easy claims.

\begin{enumerate}

\item {\em If $p_i < p$, then $v_i$ is matched in $M'$.} This follows because at the time that $M(v_i)$ arrived, either $v_i$ was already matched (as desired), or it was available for matching with $M(v_i)$ and preferable (in terms of price) over all other items that $M(v_i)$ desires (as all have price at least $p > p_i$).

\item {\em The utility of $M(v_i)$ in $M'$ satisfies $y(M(v_i)) \ge 1-p$.} This follows because under $M'_{-i}$ the utility of $M(v_i)$ is $1 - p$, and under the greedy algorithm considered, the introduction of an additional item (the item $v_i$ when considering $M'$) cannot decrease the utility of any agent. (At every step of the arrival process, the set of exposed vertices under $M'$ contains the set of exposed vertices under $M'_{-i}$, and one more vertex.)

\end{enumerate}

Using the above two claims and taking $z$ to be the value satisfying $p = e^{z-1}$, we have:

$$E[y(M(v_i)) + r(v_i)] \ge 1 - p + Pr[p_i < p]p_i= 1 - e^{z-1} + \int_{w_i=0}^{z} e^{w_i-1} dw_i = 1 - \frac{e^z}{e} + \frac{e^z - 1}{e} = 1 - \frac{1}{e}$$

This completes the proof of Lemma~\ref{lem:claims}.
\end{proof}

Using the linearity of the expectation, we have that

$$\rho_n(Ranking,G) = \sum_{v_i \in V} E[r(v_i) + y(M(v_i))] \ge (1 - \frac{1}{e})n$$

This completes the proof of Theorem~\ref{thm:KVV90a}.
\end{proof}

One can adapt the proof presented above to the special case in which the input graph is {\em MonotoneG} (or more generally, comes from the distribution $D_n$). In this case one can upper bound the slackness involved in the proof of Theorem~\ref{thm:KVV90a}, and infer the following theorem.



\begin{theorem}
\label{thm:RankingUpperbound}
For every $n$ it holds that $\rho_n(Ranking,MonotoneG) \le (1 - \frac{1}{e})n + \frac{1}{e}$.
\end{theorem}

\begin{proof}
Recall the two properties mentioned in the beginning of the proof of Theorem~\ref{thm:exact1}. Recall also that the analysis of {\em Ranking} in the proof of Theorem~\ref{thm:KVV90a} (within Lemma~\ref{lem:claims}) involved the matching $M'$ and other matchings $M'_{-i}$, and two claims. Let us analyse the slackness involved in these claims when the input is the monotone graph. The claims are restated with $M(v_i)$ replaced by $u_i$, because for the monotone graph $M(v_i) = u_i$.

The first claim stated that {\em if $p_i < p$, then $v_i$ is matched in $M'$.} When the input is the monotone graph, then a converse also holds: if $p_i > p$, then $v_i$ is not matched in $M'$. This follows because up to the time that $u_i$ arrives and is matched, only vertices of $V$ priced at most $p$ are matched, and thereafter, no other vertex in $U$ desires $v_i$. The event that $p_i = p$ has probability~0. Hence there is no slackness involved in the first claim -- it is an {\em if and only if} statement.

The second claim stated that {\em the utility of $u_i$ in $M'$ is $y(u_i) \ge 1-p$}. This inequality is not tight. Rather, the utility of $u_i$ in $M'_{-i}$ is $1-p$, and $y(u_i)$ is not smaller. Let us quantify the slackness involved in this inequality by introducing slackness variables $s(u)$.
For a vertex $u \in U$ we shall use the notation $y(u)$ to denote the utility of $u$ under {\em Ranking}, and $y_{-v}(u)$ for the utility of $u$ when vertex $v \in V$ is removed. The {\em slackness} $s(u_i)$ of vertex $u_i$ is defined as $s(u_i) = y(u_i) - y_{-v_i}(u_i)$.

\begin{lemma}
\label{lem:last}
For the monotone graph and an arbitrary vertex $u_j \in U$, the expected utility of $u_j$ (expectation taken over choices of $w_i$ for all $1 \le i \le n$ by the {\em Ranking} algorithm) is identical in the following two settings: when $v_j$ is removed, and when $v_n$ is removed. Namely, $E[y_{-v_j}(u_j)] = E[y_{-v_n}(u_j)]$.
\end{lemma}

\begin{proof}
Both $v_j$ and $v_n$ are neighbors of all vertices $u_k$ arriving up to $u_j$ (for $1 \le k \le j$). Hence whichever of the two vertices, $v_j$ or $v_n$, is removed, the distributions of the outcomes of {\em Ranking} on the first $j$ arriving vertices (including $u_j$) are the same.
\end{proof}

As a consequence of Lemma~\ref{lem:last} we deduce that for the monotone graph, the expected slackness of every vertex $u \in U$ satisfies $E[s(u)] = E[y(u)] - E[y_{-v_n}(u)]$.

\begin{lemma}
\label{lem:slackness}
For the monotone graph and arbitrary setting of prices for the items (as chosen at random by {\em Ranking}), $\sum_{u\in U} s(u) \le 1 - p_{n}$. Consequently, $\sum_{u\in U} E[s(u)] \le \frac{1}{e}$, where expectation is taken over choice of weights $w_i$ for vertices in $V$.
\end{lemma}

\begin{proof}
Fix the prices $p_i$ (hence $\pi$). Let $u_1, \ldots, u_k$ be the vertices of $U$ matched under {\em Ranking}, and let $m(u_1), \ldots, m(u_k)$ be the vertices in $V$ to  which they are matched. Observe that the prices $p(m(u_i))$ (where $1 \le i \le k$) of these vertices form a monotonically increasing sequence. Necessarily, $v_{n}$ is one of the matched vertices, because it is a neighbor of all vertices in $U$. Let $j$ be such that $v_{n} = m(u_j)$.

Consider now what happens when $v_{n}$ is removed. The vertices $u_1, \ldots, u_{j-1}$ are matched to $m(u_1), \ldots, m(u_j-1)$ as before. As to the vertices $u_j, \ldots, u_{k-1}$, they can be matched to  $m(u_{j+1}), \ldots, m(u_k)$, hence the algorithm will match them to vertices of no higher price. Specifically, for every $i$ in the range $j \le i \le k-1$, vertex $u_i$ will be matched either to $m(u_{i+1})$ or to an earlier vertex, though not earlier than $m(u_i)$. The vertex $u_k$ may either be matched or be left unmatched. For simplicity of notation, we say that $u_k$ is matched to either $m(u_{k+1})$ or to an earlier vertex, where $m(u_{k+1})$ is an auxiliary vertex of price~1 than indicates that $u_{k}$ is left unmatched.

Note that:

$$\sum_{u \in U} y(u) = \sum_{i=1}^k y(u_i) = k - \sum_{i=1}^k p(m(u_i))$$

\noindent and that:

$$\sum_{u \in U} y_{-v_{n}}(u) = \sum_{i=1}^k y_{-v_{n}}(u_i) \ge k - \sum_{i=1}^{j-1} p(m(u_i)) - \sum_{i=j+1}^{k+1} p(m(u_i))$$

Hence we have that:

$$\sum_{u\in U} s(u) = \sum_{u \in U} y(u)  - \sum_{u \in U} y_{-v_{n}}(u) \le p(m(u_{k+1})) - p(v_{n})$$

Finally, noting that $p(m(u_{k+1})) \le 1$ and that $E[p(v_{n})] = 1 - \frac{1}{e}$ (see Equation~(\ref{eq:price})), the lemma is proved.
\end{proof}

As in the proof of Theorem~\ref{thm:KVV90a} we have:

$$E[y(u_i) + r(v_i)] = 1 - p + s(u_i) + Pr[p_i < p]p_i  = 1 - \frac{1}{e} + s(u_i)$$

Using the linearity of the expectation and Lemma~\ref{lem:slackness} we have that:

$$\rho_n(Ranking,MonotoneG) = \sum_{v_i \in V} E[r(v_i) + y(u_i)] = (1 - \frac{1}{e})n + \sum_{u \in U} E[s(u)] \le (1 - \frac{1}{e})n + \frac{1}{e}$$

This completes the proof of Theorem~\ref{thm:RankingUpperbound}.
\end{proof}

\section{An alternative proof of a combinatorial identity}
\label{sec:directly}

We present a proof of Corollary~\ref{cor:a(n)} that does not make use of Lemma~\ref{lem:a(n,i)}.

\begin{proof}[Corollary~\ref{cor:a(n)}]
Let $\Pi_{\nsim (n+1)}$ denote those permutations $\pi' \in \Pi_{n+1}$ such that if {\em Ranking} uses $\pi'$ when the input is {\em MonotoneG} (with $|U|=|V| = n+1$), then $\pi'_{n+1}$ (the last item in $\pi'$) is not matched. By definition of $a(n,i)$ the expression $(n+1)! - a(n+1,n+1)$ can be interpreted as $|\Pi_{\nsim (n+1)}|$. We describe a bijection $B$ between $(\Pi_n, [n+1])$ and $\Pi_{n+1}$. The bijection will have the property that given a pair $(\pi \in \Pi_n , i \in [n+1])$ the resulting permutation $B(\pi,i) \in \Pi_{n+1}$ belongs to $\Pi_{\nsim (n+1)}$ if and only if $u_i$ is matched, thus proving the Corollary.

We now describe the bijection for a given $\pi \in \Pi_n$ and $i \in [n+1]$:

\begin{itemize}

\item $B(\pi,n+1)$: place $v_{n+1}$ at location $n+1$. This gives one permutation that we call $\pi_{\rightarrow (n+1)}$.

\item $B(\pi,i)$ for $1 \le i \le n$: place $v_{n+1}$ at location $i$ and place $v_i$ at location $n+1$. This gives $n$ additional permutations, named $\pi_{\leftrightarrow i}$ (the notation $\leftrightarrow$ indicates that $v_{n+1}$ is swapped with $v_i$).

\end{itemize}

There are three cases to consider:

\begin{itemize}

\item $i = n+1$.
In $\pi_{\rightarrow (n+1)}$ the item $v_{n+1}$ at location $n+1$ is matched, because it is a neighbor of all vertices in $U$.

\item $u_i \in U$ is matched in $\pi$. Then all vertices up to $u_i$ are also matched in $\pi_{\leftrightarrow i}$, and to items at locations no later than $n$. This is because the only differences between $\pi$ and $\pi_{\leftrightarrow i}$ involve vertices $v_i$ and $v_{n+1}$, and both of them are neighbors of all arriving vertices up to and including $u_i$. None of the vertices $u_{i+1}, \ldots, u_{n+1}$ is a neighbor of $v_i$, hence in $\pi_{\leftrightarrow i}$ the item $v_i \in V$ at location $n+1$ is not matched.

\item $u_i \in U$ is not matched in $\pi$. In this case $u_i$ will not be matched to any of the first $n$ items of $\pi_{\leftrightarrow i}$ (again, because the only differences between $\pi$ and $\pi_{\leftrightarrow i}$ involve vertices $v_i$ and $v_{n+1}$, and both of them are neighbors of all arriving vertices up to and including $u_i$). Consequently, $u_i$ will be matched to $v_i$ that is at location $n+1$ in $\pi_{\leftrightarrow i}$.

\end{itemize}

\end{proof}

\end{appendix}

\end{document}